\documentclass[12pt]{amsart}
\numberwithin{equation}{section}
\theoremstyle{plain}
\newtheorem{theorem}{Theorem}[section]
\newtheorem{lemma}[theorem]{Lemma}

\theoremstyle{definition}
\newtheorem{problem}{Problem}

\newcommand{\N}{\mathbb{N}}
\newcommand{\Z}{\mathbb{Z}}

\newcommand{\A}{\mathbf{A}}
\newcommand{\B}{\mathbf{B}}
\newcommand{\C}{\mathbf{C}}
\newcommand{\D}{\mathbf{D}}
\newcommand{\X}{\mathbf{X}}

\newcommand{\CSP}{\mathrm{CSP}}
\newcommand{\PCSP}{\mathrm{PCSP}}
\newcommand{\Pol}{\mathrm{Pol}}

\newlength{\ldprobleft} 
\newlength{\ldprobmid}  
\newlength{\ldprobright}

\newcommand{\dproblem}[3]{
\begin{equation*}
\parbox{\textwidth}{
\begin{tabular}{ @{} p{\ldprobleft} p{\ldprobmid} p{\ldprobright} @{} }
& \multicolumn{2}{l}{\textbf{#1}} \\
& {\begin{minipage}[t]{\ldprobmid}Input:\vspace{1.5pt}\end{minipage}} & {\begin{minipage}[t]{\ldprobright}#2\vspace{1.5pt}\end{minipage}} \\ 
& {\begin{minipage}[t]{\ldprobmid}Output:\end{minipage}} & {\begin{minipage}[t]{\ldprobright}#3\end{minipage}} \\
\end{tabular}}
\end{equation*}
}

\usepackage{color}

\title{Sandwiches for Promise Constraint Satisfaction}

\author[G.~Deng]{Guofeng Deng}
\author[E.~El Sai]{Ezzeddine El Sai}
\author[T.~Manders]{Trevor Manders}
\author[P.~Mayr]{Peter Mayr}
\author[P.~Nakkirt]{Poramate Nakkirt}
\author[A.~Sparks]{Athena Sparks}

\address{Department of Mathematics, University of Colorado,  Boulder, USA}
\email{guofeng.deng@colorado.edu, ezzeddine.elsai@colorado.edu,  trevor.manders@colorado.edu, peter.mayr@colorado.edu, poramate.nakkirt@colorado.edu, athena.sparks@colorado.edu}

\thanks{The authors were supported by the National Science Foundation under Grant No. DMS 1500254 and by the REU program of the Department of Mathematics at CU Boulder.}

\date{\today}

\begin{document}
\maketitle

\begin{abstract}
 Promise Constraint Satisfaction Problems ($\PCSP$) were proposed recently by Brakensiek and Guruswami~\cite{BG:PCS}
 as a framework to study approximations for Constraint Satisfaction Problems ($\CSP$).
 Informally a $\PCSP$ asks to distinguish between whether a given instance of a $\CSP$ has a solution or
 not even a specified relaxation can be satisfied.
 All currently known tractable $\PCSP$s can be reduced in a natural way to tractable $\CSP$s. 
 Barto~\cite{Ba:PMF} presented an example of a $\PCSP$ over Boolean structures for which this reduction requires 
 solving a $\CSP$ over an infinite structure.
 We give a first example of a $\PCSP$ over Boolean structures which reduces to a tractable $\CSP$ over
 a structure of size $3$ but not smaller.
 Further we investigate properties of $\PCSP$s that reduce to systems of linear equations or to $\CSP$s 
 over structures with semilattice or majority polymorphism. 
\end{abstract}

\section{Introduction}

 The Constraint Satisfaction Problem ($\CSP$) for a fixed relational structure $\A$ can be formulated as the following
 decision problem:
\dproblem{$\CSP(\A)$}{a relational structure $\X$ of the same type as $\A$}{yes, if there exists a homomorphism $\X\to\A$, \\ no, otherwise}
 For example, the question of whether a given graph is $r$-colorable is a $\CSP$ where $\A$ is the complete graph $K_r$ on $r$ vertices.
 
 In~\cite{BG:PCS} Brakensiek and Guruswami introduced Promise Satisfaction Problems ($\PCSP$) as relaxations
 and generalizations of $\CSP$. For relational structures $\A,\B$ of the same type, let
 \dproblem{$\PCSP(\A,\B)$}{a relational structure $\X$ of the same type as $\A$}{yes, if there exists a homomorphism $\X\to\A$, \\ no, if there exists no homomorphism $\X\to\B$}
 Here the \emph{promise} is that for the input $\X$ exactly one of the two alternatives $\exists\, \X\to\A$ or
 $\not\exists\, \X\to\B$ holds.
 A typical example of a $\PCSP$ is to distinguish graphs that are $r$-colorable from those that are not even
 $s$-colorable for $r\leq s$.
 
 Let $\A,\B,\C$ be relational structures of the same type with homomorphisms $\A\to\C\to\B$. Then we say $\C$ is 
 \emph{sandwiched} by $\A$ and $\B$.\footnote{Imagine $\C$ as the cheese between avocado $\A$ and bread $\B$ in
 the sandwich.}
 In this case $\PCSP(\A,\B)$ has a straightforward (constant time) reduction to $\CSP(\C)$:
 a structure $\X$ is a yes-instance for $\PCSP(\A,\B)$ iff $\X$ is a yes-instance for $\CSP(\C)$.
 In general, the complexity of $\PCSP$ is unknown. However all currently known tractable $\PCSP(\A,\B)$ can be
 reduced to tractable $\CSP(\C)$ for some $\C$ sandwiched by $\A$ and $\B$. 
 
 In a research project for undergraduate students (REU) organised by P.~Mayr and A.~Sparks at CU Boulder in Summer
 2019, we considered the following {\bf meta question} on $\PCSP$:
\begin{quote}
 Given finite $\A,\B$, does there exists some $\C$ sandwiched by $\A$ and $\B$ such that $\CSP(\C)$ is tractable?
\end{quote}
 If the answer is yes, then clearly $\PCSP(\A,\B)$ is tractable. However a negative answer may not necessarily
 yield hardness of  $\PCSP(\A,\B)$.
 
 In any case it is not known whether the meta question is decidable. The main obstacle is that tractable sandwiched
 structures may grow in size.
 Barto gave an example of Boolean $\A,\B$ for which all tractable sandwiched $\C$ are infinite~\cite{Ba:PMF}.
 Moreover, it is open whether the size of the smallest finite sandwiched tractable $\C$,
 that is, the function
\[ c(\A,\B) := \min\{ |\C|\ :\ \A\to\C\to\B, \C \text{ finite}, \CSP(\C) \text{ tractable} \} \]
 is computable (If no such $\C$ exists, let $c(\A,\B)$ be undefined). 
 One outcome of the REU is a first example of Boolean $\A,\B$ for which the smallest sandwiched $\C$ with
 tractable $\CSP(\C)$ has size $3$ (see Theorem~\ref{thm:example}). 
 In particular $c(\A,\B)$ is not bounded above by $\max(|A|,|B|)$. 

 In Section~\ref{sec:majority} we show that if $\A,\B$ sandwich some $\C$ with a conservative polymorphism
 (or a majority polymorphism in case $|A| =2$), then they sandwich some $\D$ with the same polymorphism of size
 $|D|\leq |A|$.

\section{Affine sandwiches} 

 A relational structure $\C$ is \emph{affine} if its domain $C$ forms an abelian group $\mathbb{C} := (C,+,-,0)$
 and $x-y+z$ is a polymorphism of $\C$, that is, in
 \[ \Pol(\C) := \{ f\colon \C^k\to \C \ :\ k\in\N \}. \]
 In other words, $\C$ is affine if every $n$-ary relation $R^\C$ of $\C$ is a coset of a subgroup of
 $\mathbb{C}^n$. Then $\CSP(\C)$ encodes a system of linear equations and can be solved in polynomial time.

 We present an example of Boolean $\A,\B$ with sandwiched affine $\C$ of size $3$ but without any sandwiched tractable
 Boolean structure.

\begin{theorem}  \label{thm:example}
 Let $\A = (\{0,1\}, R^\A),\B = (\{0,1\}, R^\B),\C = (\{0,1,2\}, R^\C)$  with
\begin{align*}
 R^\A & := \{ 100011, 010101, 001110 \}, \\
 R^\B & := \{0,1\}^6 \setminus \{ 000000,000111,111000,111111 \}, \\
 R^\C & := \text{the closure of }  R^\A \text{ under } x-y+z \bmod 3. 
\end{align*}
 Then
\begin{enumerate}
\item \label{it:sandwich}
 the affine $\C$ is sandwiched by $\A$ and $\B$ via homomorphisms
 $\A \xrightarrow{\mathrm{id}} \C \xrightarrow{g} \B$ where $g\colon\{0,1,2\} \to \{0,1\}$ is defined by
 $g(0)=g(2)=0$ and $g(1)=1$, 
\item \label{it:noBoolean} 
 but there exists no Boolean $\D$ such that $\A\to\D\to\B$ and $\CSP(\D)$ is tractable
 (assuming $\mathrm{P}\neq\mathrm{NP}$).
\end{enumerate}
\end{theorem}

\begin{proof}
 For~\eqref{it:sandwich} note first that $\A$ is a substructure of $\C$ by definition, that is, the identity map
 $\A\to\C$ is a homomorphism.
 More explicitly $R^\A$ spans the affine subspace
\begin{align*}  
  R^\C & =  \{ (x_1,\dots,x_6)\in\Z_3^6 \ :\ \begin{aligned}[t] & x_1+x_2+x_3 = 1, \\
  & x_1+x_4 = 1, x_2+x_5= 1, x_3+x_6 = 1 \}
\end{aligned} \\    
  & = \{ \begin{aligned}[t] & 100011, 010101, 001110, \\
          & 220221, 202212, 022122, \\
          & 211200,121020, 112002 \}.
\end{aligned}          
\end{align*}
 Applying $g$ we get
\begin{align*}
 g(R^{\mathbf{C}}) & = \{ \begin{aligned}[t] & 100011, 010101, 001110, \\
          & 000001, 000010, 000100, \\
          & 011000,101000, 110000 \} \subseteq R^\B.
\end{aligned}          
\end{align*}
 Hence $g\colon \C\to\B$ is a homomorphism, and~\eqref{it:sandwich} is proved.

 For~\eqref{it:noBoolean} suppose there exists Boolean $\D$ and homomorphisms $\A\xrightarrow{f}\D\xrightarrow{h}\B$
 such that $\CSP(\D)$ is tractable. Since $R^\B$ contains no constant tuple, both $f$ and $h$ are bijections.

 {\bf Case, $f$ and $h$ both are the identity:} Then $\A\leq\D\leq\B$ is a chain of substructures.
 By Schaefer's Dichotomy Theorem for Boolean $\CSP$~\cite{Sc:CSP}, $\D$ has one of the following polymorphisms:
 $0,1,\wedge,\vee$, minority or majority.
 However, closing $R^\A$ under any of the first 4 clearly yields a constant tuple, e.g.,
\[ 100011 \wedge 010101 \wedge 001110 = 000000, \]
 which is not in $R^\B$. Hence $R^\D$ is not preserved by
 $0,1,\wedge,\vee$. Next applying the minority operation $d$ to the elements of $R^\A \subseteq R^\D$ yields, e.g.,
\[ d(100011, 010101, 001110 ) = 111000 \]
 which is not in $R^\B$, hence not in $R^\D$.
 Similarly applying the majority operation $m$ yields
\[ m(100011, 010101, 001110 ) = 000111\not\in R^\D. \]
 Hence no substructure of $\B$ that contains $\A$ has tractable $\CSP$.

 {\bf Case, $f$ is the identity and $h$ is negation:} As in the previous case, closing $R^\A$ under one of the
 six polymorphisms of Schaefer's Theorem and then applying $h$ to the results yields a constant tuple, $111000$,
 or $000111$. Since neither is in $R^\B$, we have a contradiction.

 The remaining cases that $f$ is negation, $h$ is the identity or that both $f$ and $h$ are negation follow similarly.
 Thus~\eqref{it:noBoolean} is proved.
\end{proof}

 There is no known characterization of structures $\A,\B$ that sandwich an affine $\C$. But we have some 
 necessary conditions in terms of the polymorphisms from $\A$ to $\B$,
\[ \Pol(\A,\B) := \{ \A^k\to\B \ :\ k\in\N \}. \] 
 First recall that for $k\in\N$ a function $f\colon A^k\to B$ is \emph{symmetric} if $f$ is invariant under
 permutation of its arguments $x_1,\dots,x_k\in A$.
 More generally Brakensiek and Guruswami~\cite{BG:SPED} call $f\colon A^k\to B$ \emph{block-symmetric}
 for a partition of $\{1,\dots,k\}$ into blocks $B_1,\dots,B_\ell$ if $f$ is invariant under permutation of
 arguments $x_{i_1},\dots,x_{i_m}$ for any block $B_j = \{i_1,\dots,i_m\}$. Note that every function $f$ is
 block-symmetric for the partition into singletons. Further there exists a unique coarsest partition
 for which $f$ is block-symmetric, that is, a partition with maximal blocks $B_1,\dots,B_\ell$. 
 The \emph{width} of $f$ is the size of the smallest block of this coarsest partition for which $f$ is block
 symmetric, that is,
\[ \max\{ \min\{ |B_1|,\dots,|B_\ell| \} \ :\ f \text{ is block-symmetric for the partition } B_1,\dots,B_\ell \}. \]
 We can now formulate some weak necessary conditions for structures to have an affine sandwich.
 Clearly they are not sufficient.
 
\begin{theorem} \label{thm:affine}
 Let $\A,\B,\C$ be relational structures of the same type with homomorphisms $\A\to\C\to\B$ and $\C$ affine.
\begin{enumerate}
\item \label{it:bs}
 Then $\Pol(\A,\B)$ contains block-symmetric polymorphisms of arbitrary large width.
\item \label{it:fs}
 If $\C$ is finite, then $\Pol(\A,\B)$ contains symmetric polymorphisms of arbitrary large arity.
\end{enumerate}
\end{theorem}

\begin{proof}
 Let $f\colon\A\to\C$ and $g\colon\C\to\B$. Since $\C$ is affine, we have for each $k\in\N$ and
 $a_1,\dots,a_k\in\Z$ such that $\sum_{i=1}^k a_i x = x$ for all $x\in C$, that 
\[ C^k\to C, (x_1,\dots,x_{k}) \mapsto \sum_{i=1}^{k} a_i x_i,  \]
 is in $\Pol(\C)$. Composing this polymorphism with the homomorphisms $f$ and $g$, we obtain that
\[ A^k\to B, (x_1,\dots,x_{k}) \mapsto g\left( \sum_{i=1}^{k} a_i f(x_i) \right), \]
 is in $\Pol(\A,\B)$.

 For~\eqref{it:bs}, it follows that for $k\in\N$
\[ A^{2k+1}\to B,\ (x_1,\dots,x_{2k+1}) \mapsto g\left( \sum_{i=1}^{2k+1} (-1)^{i-1} f(x_i) \right), \] 
 is a block-symmetric polymorphism with partition into two blocks $B_1 = \{1,3,\dots,2k+1\}$, $B_2 = \{2,4,\dots,2k\}$,
 hence width $\geq k$.

 For~\eqref{it:fs} assume $\C$ is finite of size $n$. Then for $k\in\N$
\[ A^{nk+1}\to B,\ (x_1,\dots,x_{nk+1}) \mapsto g\left( \sum_{i=1}^{nk+1} f(x_i) \right), \]
 is a symmetric polymorphism of arity $nk+1$.
\end{proof}

 Assume $\A$ and $\B$ sandwich an affine $\C$. Then $\PCSP(\A,\B)$ reduces to the linear system $\CSP(\C)$.
 More generally, Brakensiek and Guruswami showed that if $\Pol(\A,\B)$ contains (block)-symmetric polymorphisms
 of arbitrary large arity (width), then $\PCSP(\A,\B)$ can be solved in polynomial time via the so-called basic
 linear programming relaxation over the non-negative rationals and over the integers~\cite[Theorem 3.1, 4.1]{BG:SPED}.

\section{Conservative and majority sandwiches} \label{sec:majority}

 We add some straightforward observations on non-affine sandwiches.
 A function $f\colon A^k\to A$ is \emph{conservative} if $f(a_1,\dots,a_k)\in\{a_1,\dots,a_k\}$ for all
 $a_1,\dots,a_k\in A$. For example, semilattice operations are conservative.

 For a structure $\C$ and $D\subseteq C$ the \emph{induced substructure} $\C|_D$ on $D$ has domain $D$ and
 relations $R^\D := R^\C \cap (D\times\dots\times D)$ for every $R$ in the type of $\C$.
 
\begin{lemma} \label{lem:conservative}
 Let $\A,\B,\C$ be relational structures of the same type with homomorphisms $\A\xrightarrow{f}\C\xrightarrow{g}\B$,
 and let $\D := \C|_{f(A)}$. 
 Then we have homomorphisms $\A\xrightarrow{f}\D\xrightarrow{g|_{f(A)}}\B$, and every conservative polymorphism of
 $\C$ restricts to a  polymorphism of $\D$. 
\end{lemma}

\begin{proof}
 Let $p$ be a conservative polymorphism of $\C$. Then $p(D,\dots,D)\subseteq D$.
 Hence $p$ preserves $R^\D$ for every relation $R$ in the type of $\C$. 
\end{proof}
 
 As a consequence, if $\A$ and $\B$ sandwich a structure with conservative Taylor polymorphism (e.g,
 a semilattice polymorphism), then they sandwich such a structure of size $\leq|A|$.
 Hence given finite $\A,\B$ it is decidable whether they sandwich some structure with conservative
 Taylor polymorphism.   

\begin{lemma} \label{lem:majority}
 Let $\A,\B,\C$ be relational structures of the same type with homomorphisms $\A\xrightarrow{f}\C\xrightarrow{g}\B$.
 Assume that $\A$ is Boolean and $\C$ has a majority polymorphism $m$.
 Then $\D := \C|_{f(A)}$ is sandwiched by $\A$ and $\B$, has a majority polymorphism $m|_{f(A)}$, and has size $\leq 2$.
\end{lemma}

\begin{proof}
 Clearly $f$ reduces to a homomorphism from $\A$ into $\D$, and $g$ restricts to a homomorphism from $\D$ to $\B$. 
 Note that $m(f(A),f(A),f(A)) \subseteq f(A)$ since $|f(A)| \leq 2$ and $m$ is only ternary. Hence $m|_{f(A)}$
 is a majority polymorphism on $\D$.
\end{proof}

 As a consequence, if some Boolean $\A$ and finite $\B$ sandwich a structure with majority polymorphism,
 then they sandwich such a structure of size $\leq 2$.
 In particular, given Boolean $\A,\B$ it is decidable whether they sandwich some structure with majority
 polymorphism.

\section{Summary}

 We gave some weak necessary conditions for finite structures $\A,\B$ to sandwich a finite affine $\C$
 and showed that the smallest such $\C$ can be strictly larger than $\A,\B$. The following remains open:

\begin{problem}
 Given finite $\A,\B$, is it decidable whether they sandwich some (finite) affine $\C$?  
\end{problem}


\end{document}